\documentclass[journal]{IEEEtran}
\usepackage[utf8]{inputenc}
\usepackage{amsmath}
\usepackage{amsthm}
\usepackage{amssymb}
\usepackage{float}
\usepackage{braket}
\usepackage{graphicx}
\usepackage{url}
\usepackage{here}
\usepackage{rotating}
\usepackage[colorlinks=true, pdfstartview=FitV, linkcolor=red, citecolor=blue, urlcolor=blue]{hyperref}
\usepackage{slashed}
\usepackage{multirow}
\usepackage{tcolorbox}
\usepackage{subcaption}
\usepackage{orcidlink}
\usepackage{algorithm}
\usepackage{algpseudocode}

\usepackage[normalem]{ulem}
\graphicspath{{./figures/}}

\newcommand{\be}{\begin{equation}}      
\newcommand{\ee}{\end{equation}}      
\newcommand{\bea}{\begin{eqnarray}}      
\newcommand{\eea}{\end{eqnarray}}

\newcommand{\tr}{\mathrm{tr}}

\newcommand{\Tr}{\mathrm{Tr}}

\newcommand{\ctext}[1]{\raise0.2ex\hbox{\textcircled{\scriptsize{#1}}}}

\usepackage{tikz}
\usetikzlibrary{quantikz}

\newtheorem{theorem}{Theorem}[section]

\newtheorem{proposition}[theorem]{Proposition}

\theoremstyle{definition}

\theoremstyle{remark}

\usepackage[normalem]{ulem}

\begin{document}

\title{Quantum interactive proofs using quantum energy teleportation}


\author{Kazuki Ikeda\,\orcidlink{0000-0003-3821-2669} and~Adam Lowe \orcidlink{0000-0002-3714-4193}
\thanks{K. Ikeda is with Co-design Center for Quantum Advantage (C2QA) and Center For Nuclear Theory, Department of Physics and Astronomy, Stony Brook University, USA e-mail:kazuki.ikeda@stonybrook.edu and kazuki7131@gmail.com website: \url{https://kazukiikeda.studio.site/}}
\thanks{A. Lowe is with College of Engineering and Physical Sciences, School of Informatics and Digital Engineering, Aston University, Birmingham B4 7ET, United Kingdom e-mail:a.lowe3@aston.ac.uk.}}
\maketitle

\begin{abstract}
We present a simple quantum interactive proof (QIP) protocol using the quantum state teleportation (QST) and quantum energy teleportation (QET) protocols. QET is a technique that allows a receiver at a distance to extract the local energy by local operations and classical communication (LOCC), using the energy injected by the supplier as collateral.
QET works for any local Hamiltonian with entanglement and, for our study, it is important that getting the ground state of a generic local Hamiltonian is quantum Merlin Arthur (QMA)-hard. The key motivations behind employing QET for these purposes are clarified. Firstly, in cases where a prover possesses the correct state and executes the appropriate operations, the verifier can effectively validate the presence of negative energy with a high probability (Completeness). Failure to select the appropriate operators or an incorrect state renders the verifier incapable of observing negative energy (Soundness). Importantly, the verifier solely observes a single qubit from the prover's transmitted state, while remaining oblivious to the prover's Hamiltonian and state (Zero-knowledge). Furthermore, the analysis is extended to distributed quantum interactive proofs, where we propose multiple solutions for the verification of each player's measurement. The results in the $N$-party scenario could have particular relevance for the implementation of future quantum networks, where verification of quantum information is a necessity. The complexity class of our protocol in the most general case belongs to QIP(3)=PSPACE, hence it provides a secure quantum authentication scheme that can be implemented in small quantum communication devices. It is straightforward to extend our protocol to Quantum Multi-Prover Interactive Proof (QMIP) systems, where the complexity is expected to be more powerful (PSPACE$\subset$QMIP=NEXPTIME). In our case, all provers share the ground state entanglement, hence it should belong to a more powerful complexity class QMIP$^*$.
\end{abstract}

\section{Introduction and Summary of Our Results}
Quantum interactive proofs play a crucial role in verifying and validating quantum information in various applications. In this paper, we present a simple and robust QIP system that leverages the quantum state teleportation ~\cite{PhysRevLett.70.1895,furusawa1998unconditional,2015NaPho...9..641P,takeda2013deterministic} and quantum energy teleportation protocols~\cite{HOTTA20085671,2009JPSJ...78c4001H,2015JPhA...48q5302T,2023arXiv230111884I,Ikeda:2023uni,Ikeda:2023xmf,PhysRevD.107.L071502}. QET, a technique enabling the extraction of local energy at a distance through LOCC, utilizes injected energy as collateral. The protocol works for any local Hamiltonian with entanglement, where determining the ground state of a generic local Hamiltonian is known to be QMA-hard,~\cite{10.1007/978-3-540-30538-5_31,broadbent2022qma} even for 1d spin chain Hamiltonians with nearest-neighbor interactions~\cite{Aharonov2009}. The protocols of QET have been rigorously proven and are being studied from various physical points of view, such as efficiently simulating phase transitions in quantum many-body systems~\cite{Ikeda:2023xmf,PhysRevD.107.L071502}. However, since the energy that can be transferred by QET is extremely small, it is better suited for use in a quantum way than in a classical way. 

In the previous paper by one of the authors, QET was extended to large-scale, long-distance quantum networks in combination with QST~\cite{2023arXiv230111884I}. The primary objective of this study is to elucidate the key motivations behind incorporating QET into systems which use QIPs and subsequently to utilize QET as a new authentication method on quantum networks.  Firstly, by employing QET, the verifier can effectively ascertain the presence of negative energy with a high probability when the prover possesses the correct state and performs the appropriate operations (Completeness). Conversely, failure to select the appropriate operators or an incorrect state renders the verifier unable to observe negative energy (Soundness). This characteristic ensures the integrity and accuracy of the verification process. Moreover, a system which utilises QIPs, guarantees a zero-knowledge property where the verifier solely observes some qubits sent by the prover using QST. Importantly, the verifier remains oblivious to both the prover's Hamiltonian and ground state, ensuring the confidentiality of sensitive information.

In addition to the QIP between one prover and one verifier, we extend our analysis to distributed QIPs, introducing multiple verifiers who verifies the prover's system. This aspect is particularly relevant in the context of $N$-party QET, as it holds promise for the implementation of future quantum networks, where the verification of quantum information is of utmost importance. Our method can be applied to establish a secure consensus building among $N$-parties~\cite{2022arXiv221102073I} and various distributed quantum systems~\cite{IKEDA2018199,10.1007/978-3-030-01174-1_58}.

The organization of this article and our original contributions are summarized as follows. First, we describe the most common protocols for QIP with QET (Sec.~\ref{sec:QIP}). Then, for the simplest case, we present an exact and optimized quantum circuit of QIP using the minimal QET model (Sec.~\ref{sec:minimal}), and examine Soundness, Zero-knowledge, and Completeness in detail. We also present a protocol that extends two-party QIP to distributed QIP with any number of verifiers. In comparison to existing secure authentication protocols, our method offers a relatively straightforward approach that can be readily implemented on current quantum computers and quantum networks. 
To validate the efficacy of our protocol, we conducted rigorous test using both IBM Qiskit and exact diagonalization of the Hamiltonian. The results not only confirmed its consistency with theoretical predictions but also demonstrated its compatibility with Noisy Intermediate-Scale Quantum (NISQ) devices\footnote{Readers can test our algorithm using a real quantum computer from the link \url{https://github.com/IKEDAKAZUKI/Quantum-Energy-Teleportation}}. By combining our approach with conventional quantum state teleportation, simultaneous authentication becomes feasible, without being restricted by the number or distance of authentication targets. Furthermore, our method is format-agnostic, accommodating both centralized and decentralized database structures utilized by the parties involved. Finally, we propose a new Quantum Multi-Prover
Interactive Proof (QMPI$^*$) with entangled provers by introducing many energy suppliers, who act as provers of the system (Sec.~\ref{sec:QMIP}). The complexity is truly more powerful than for QIPs, hence it is ultimately secure. Thus, it is an ultimately secure authentication system that cannot be broken even with the best conceivable quantum devices.

\section{Important Related Works on QIPs}
Watrous~\cite{814583} introduced QIPs in the centralized scenario as a variation of classical Interactive Proofs (IP). In QIPs, the verifier can perform polynomial-time quantum computation instead of classical computation, and the prover and verifier can exchange quantum bits instead of classical bits. Kitaev and Watrous~\cite{10.1145/335305.335387} demonstrated that QIPs, which comprises languages decidable by a quantum interactive protocol with a polynomial number of interactions, is included in EXP, the class of languages decided in exponential time. They also demonstrated that any QIP protocol with a polynomial number of interactions can be parallelized into three turns (QIP = QIP(3)). Here QIP(3) is a protocol in which messages are exchanged three times between the prover and the verifier.

Jain, Ji, Upadhyay, and Watrous~\cite{10.1145/2049697.2049704} further improved this containment by showing that a QIP is actually contained in PSPACE, implying that a QIP collapses to the complexity class IP (QIP = IP = PSPACE). Although this result establishes that quantum interactive proofs are not more powerful than classical interactive proofs, there is an intriguing property unique to quantum interactive proofs: the number of interactions can be significantly reduced and protocols can be simpler. Watrous initially proved that any language in PSPACE can be decided by a three-turn QIP protocol~\cite{PhysRevA.65.012310}. Marriott and Watrous~\cite{Marriott2005} additionally showed that the verifier's turn in QIP(3) protocols can be replaced by a classical 1-bit (QIP(3) = QMAM) and this is indeed our case. Moreover distributed QIP~\cite{gall2022distributed} was proposed by developing classical IP~
\cite{55f26229cb9e4e719ad6aa994bc26328}.

The extension of an IP to the case where there are multiple provers is called a MIP~\cite{10.1145/62212.62223}. The complexity of MIPs is equivalent to NEXPTIME~\cite{Babai1991}, therefore it is expected to be more powerful than QIP. Moreover a quantum version of MIPs is called a QIP, where provers do not share entanglement and use shared random classical bits. The complexity is equivalent to QMIP=MIP= NEXPTIME~\cite{10.1007/3-540-36136-7_11,doi:10.1137/090751293}. When provers share the entanglement, a MIP is extended to MIP$^*$, which is proven not equal to a MIP hence it is truly more powerful than a MIP (MIP$\subsetneq$MIP$^*$)~\cite{2019arXiv190405870N}. QMIP with provers that share the entanglement state is called a QMIP$^*$ and is known to be as powerful as a MIP$^*$~\cite{reichardt2013classical}.

\section{\label{sec:QIP}Local Hamiltonian problem and quantum energy teleportation}
\subsection{General protocol of quantum energy teleportation}
We first explain the QET protocol in the most general manner. For quantum circuit implementations of particular problems, please see ~\cite{2023arXiv230111884I,Ikeda:2023uni,Ikeda:2023xmf,PhysRevD.107.L071502} and Fig.~\ref{fig:circuit}. We consider a local Hamiltonian as $H=\sum_{n=0}^{N-1}H_n$,
where $H_n$ is the local Hamiltonian interacting with some neighboring qubits and obeys the constraint
\begin{align}
\begin{aligned}
\label{eq:condition}
\bra{g}H\ket{g}=\bra{g}H_n\ket{g}=0,~\forall n\in \{1,\cdots,N\},
\end{aligned}    
\end{align}
where $\ket{g}$ is the ground state of the full Hamiltonian $H$. Note that $\ket{g}$ is not always the ground state of local $H_n$. It is important to note that $\ket{g}$ is an entangled state in general. One can always add or subtract constant values to maintain the constraint~\eqref{eq:condition}. Moreover, since $\ket{g}$ is the ground state, it is crucial that any non-trivial (local) operations to $\ket{g}$ including measurement increase the energy expectation value.

In what follows we describe the QET protocol. Let Alice be a supplier of energy and Bob be a receiver. Alice measures her Pauli operator $\sigma_{A}$, using her projective measurement operator $P_{A}(\mu)=\frac{1}{2}(1+\mu \sigma_{A})$ and obtains either $\mu=-1$ or $\mu=+1$. Initially, the injected energy $E_A$ is localized around subsystem $A$, but Alice cannot extract it from the system solely through her operations at $A$. However using LOCC, Bob can extract some energy from his local system. 

By classical communication, Alice sends her measurement result $\mu$ to Bob, who performs a conditional operation $U_{B}(\mu)$ to his state and measures his local Hamiltonian $H_{B}$. He can use
\begin{equation}
\label{eq:operation}
    U_{B}(\mu)=\cos\theta I-i\mu\sin\theta\sigma_{B},
\end{equation}
where $\theta$ obeys 
\begin{align}
\label{eq:theta}
    \cos(2\theta)=\frac{\xi}{\sqrt{\xi^2+\eta^2}},~
    \sin(2\theta)=-\frac{\eta}{\sqrt{\xi^2+\eta^2}},
\end{align}
where
\begin{align}
\label{eq:params}
\xi=\bra{g}\sigma_{B}H\sigma_{B}\ket{g},~\eta=\bra{g}\sigma_{A}\dot{\sigma}_{B}\ket{g},
\end{align}
with $\dot{\sigma}_{B}=i[H,\sigma_{B}]$. The local Hamiltonian should satisfy $[H,\sigma_{B}]=[H_{B},\sigma_{B}]$. 
The average quantum state $\rho_\text{QET}$, which is a mixed state, is obtained after Bob operates $U_{B}(\mu)$ to $\frac{1}{\sqrt{p(\mu)}}P_{A}(\mu)\ket{g}$, where $p(\mu)$ is the normalization factor.

Then the density matrix $\rho_\text{QET}$ after Bob operates $U_{B}(\mu)$ to $P_{A}(\mu)\ket{g}$ is 
\begin{equation}
\label{eq:rho_QET}
    \rho_\text{QET}=\sum_{\mu\in\{-1,1\}}U_{B}(\mu)P_{A}(\mu)\ket{g}\bra{g}P_{A}(\mu)U^\dagger_{B}(\mu). 
\end{equation}
The expected energy at Bob's local system can be evaluated as 
\begin{equation}
\label{eq:QET}
    \langle E_{B}\rangle=\Tr[\rho_\text{QET}H_{B}]=\frac{1}{2}\left[\xi-\sqrt{\xi^2+\eta^2}\right], 
\end{equation}
which is negative if $\eta\neq 0$. It is expected that if there is no energy dissipation, the positive energy of $-\langle E_{B}\rangle$ is transferred to Bob's device after the measurement due to energy conservation. 

Here let us give some remarks on QET. Although QET bears a conceptual resemblance to QST, it is crucial to highlight that in QET, it is classical information, rather than energy, that is transmitted. Additionally, the intermediate subsystem situated between Alice and Bob does not undergo excitation by the energy carriers of the system during the brief duration of the QET process. As a result, the time required for energy transport through QET is significantly shorter than the time it takes for heat generation in the system's natural time progression. Through operations performed on his local system, Bob can extract energy from a system based on the classical information conveyed by Alice. When we refer to energy teleportation or energy transfer, we mean that Bob can receive energy much more rapidly (at the speed of light) than the energy can be transmitted from Alice to Bob in the natural course of the system's evolution.

\subsection{Quantum Interactive Proof using QET}
\begin{figure}[H]
    \centering
    \includegraphics[width=\linewidth]{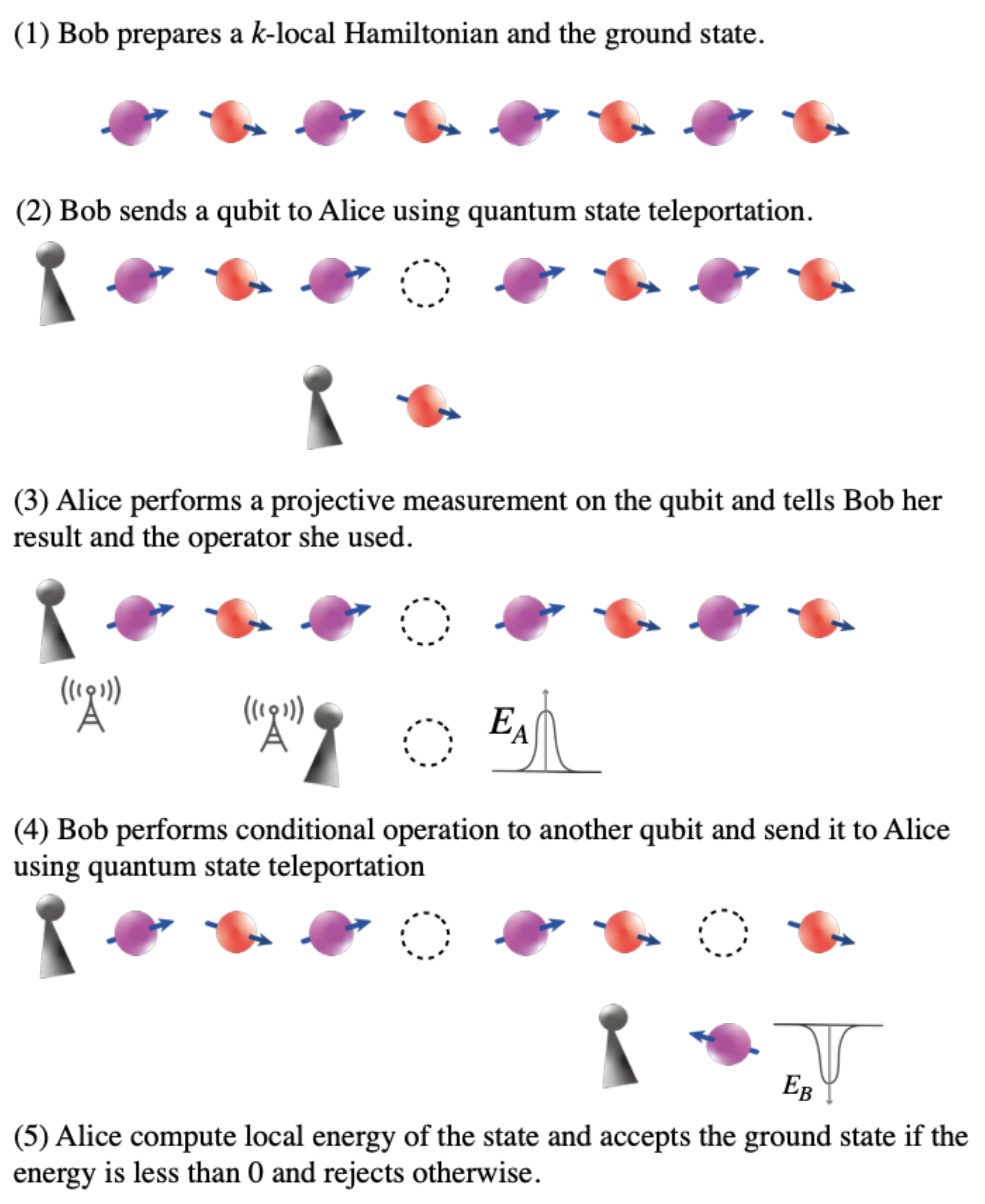}
    \caption{Schematic picture of quantum interactive proof by using both QET and QST.}
    \label{fig:QIP_QET}
\end{figure}

The prover can encode $(k_1,\cdots,k_N,h_1,\cdots,h_{N-1})$ into the ground state of a Hamiltonian, for example,
\begin{equation}
    H=\sum_{i=1}^Nk_iZ_i+\sum_{j=1}^{N-1}h_jX_jX_{j+1}.
\end{equation}
We will address the simplest case later. 

\begin{algorithm}
\caption{QIP by Quantum Energy Teleportation}\label{alg:cap}
\hspace*{\algorithmicindent} \textbf{Input} The ground state $\ket{g}$ of a Hamiltonian~\eqref{eq:condition}\\
    \hspace*{\algorithmicindent} \textbf{Output} Accept or Reject
\begin{algorithmic}
\Require $n_\text{shot}>0$
\For{$n <n_\text{shot}$}
\State Perform the projective measurement $P_{A}(\mu)$ to $\ket{g}$. 
\If{$\mu=-1$} 
    \State Apply $U_B=U(-1)$ to the state
\Else
\State Apply $U_B=U(+1)$ to the state
\EndIf
\State Measure $H_{B}$
\EndFor \Comment{$\langle E_{B}\rangle=\Tr[\rho_\text{QET}H_{B}]$ can be obtained}
\If{$\langle H_{B}\rangle<0$}
\State Accept $\ket{g}$
\Else
\State Reject $\ket{g}$
\EndIf
\end{algorithmic}
\end{algorithm}
The schematic picture of a QIP using QET is shown in Fig.~\ref{fig:QIP_QET} and the general procedure is explained as follows. 
\begin{enumerate}
\item[Step 1:] Bob creates a ground state and sends a local state $\ket{g_{A}}$ (witness) to Alice using QST.
\item[Step 2:] Alice applies a projective measurement to $\ket{g_{A}}$ and tells Bob her measurement result $\mu$ and her operator $\sigma_{A}$. 
\item[Step 3:] Bob operates $U_{B}(\mu)$ to a local state $\ket{g_{B}}$ and sends $\ket{g_{B}}$ to Alice using QST. 
\item[Step 4:] Alice measures the state $\ket{g_{B}}$ in the basis specified by Bob.
\item[Step 5:] She accepts the state $\ket{g}$ if she observes negative energy $\langle E_B\rangle<0$ on her local system and otherwise rejects it. 
\end{enumerate}

Here we give brief remarks on the protocol. The major difference from the conventional QET protocol is that Alice and Bob exchange local states twice by QST. In the conventional protocol, Alice is the supplier of energy and Bob is the receiver of energy, but in our new protocol, Alice injects energy in Step 2 and receives energy in another qubit in Step 5. This time, Bob does not receive any energy but only performs conditional operations according to the results of Alice's projective measurements. Alice and Bob need to repeat the process of Steps 1 through 5 several times, but the total circuit length required for two quantum teleportations is about 10, so the protocol is completed instantaneously for practical purposes. For example, in IBM's superconducting quantum computer, the duration of a single qubit operation is $O(10)~ns$ and the CNOT duration is $O(100)~ns$. Since only four CNOTs are used in two QSTs (see Fig.~\ref{fig:circuit} for the circuit), the duration of the entire process is $O(500)~ns$, and even if the process is repeated 1000 times, it would only take $O(500)~\mu s$, which is quite fast. Moreover, except for the QST part, Alice does not need to have a universal quantum computer herself to perform this protocol, since it requires only single qubit operations and measurement. Furthermore, note that traditional key exchange is not performed during the process.  Finally as can be seen from the process (Steps 1 to 5), the prover and the verifier exchange the message three times, but the verifier sends only the classical one bit. This indeed shows our protocol belongs to QIP(3). In fact, when the verifier checks the witness, the values $\mu\in\{-1,+1\}$ obtained by measurement are completely random each time and the verifier makes $\mu$ publicly available. Even if the prover has ultimate quantum computational resources, it is impossible to predict in advance which values of $\mu$ will appear. This is the same as in the Arthur–Merlin protocol~\cite{10.1145/22145.22192}.

\begin{tcolorbox}
[colframe = blue!50!black, 
    colback = white, 
    title = Objective of the protocol]
Bob (prover) encodes a tuple of numbers $k=(k_1,k_2,\cdots,k_m)$ into the quantum state $\ket{g(k)}$ of a local Hamiltonian $H=\sum_{n}H_n$ and wants to prove $\ket{g}$ is the genuine ground state without revealing any information about $\ket{g(k)}$ to Alice (verifier). 
\end{tcolorbox}

To avoid misunderstanding, we emphasize that, in this study, QET is treated from a viewpoint of quantum information engineering and cryptography, so we do not discuss the actual physical system or its physical processes in a rigorous manner. For us, the important thing is that the algorithms presented in this study can be implemented in real quantum devices.

Here we explain the reasons why QET can be used for QIPs and ground-state authentication. First of all, in order for Alice to be able to observe the negative energy, Bob needs to find the correct $\theta$ in eq.\eqref{eq:theta}, which is information that only Bob knows, by using eq.~\eqref{eq:params}. If Bob has a correct state $\ket{g}$ and performs the correct operations, then Alice can verify negative energy with a high probability (\textbf{Completeness}). Second, Alice can use any operator $\sigma_{nA}$ for her projective measurement, and the value of $\mu$ that Alice detects is completely random each time. Third, in the process of QET, Bob must continue to use the same $\theta$ and $\ket{g}$ (therefore Alice also uses the same $\sigma_{nA}$), and if the appropriate $\theta$ is not chosen or $\ket{g}$ is not the ground state, Alice cannot observe negative energy (\textbf{Soundness}). Fourth, since finding the ground state of a general $k$-local Hamiltonian is QMA-hard, if Alice can observe negative energy, it guarantees with a very high probability that $\ket{g}$ is the correct ground state of the Hamiltonian. Finally, Alice only observes a single qubit of the state sent by Bob, and she never learns anything about Bob's Hamiltonian and $\ket{g}$ (\textbf{Zero-knowledge}). Therefore a quantum zero-knowledge proof of Bob's system can be realized by QET.

\begin{figure*}
    \centering
    \includegraphics[width=\linewidth]{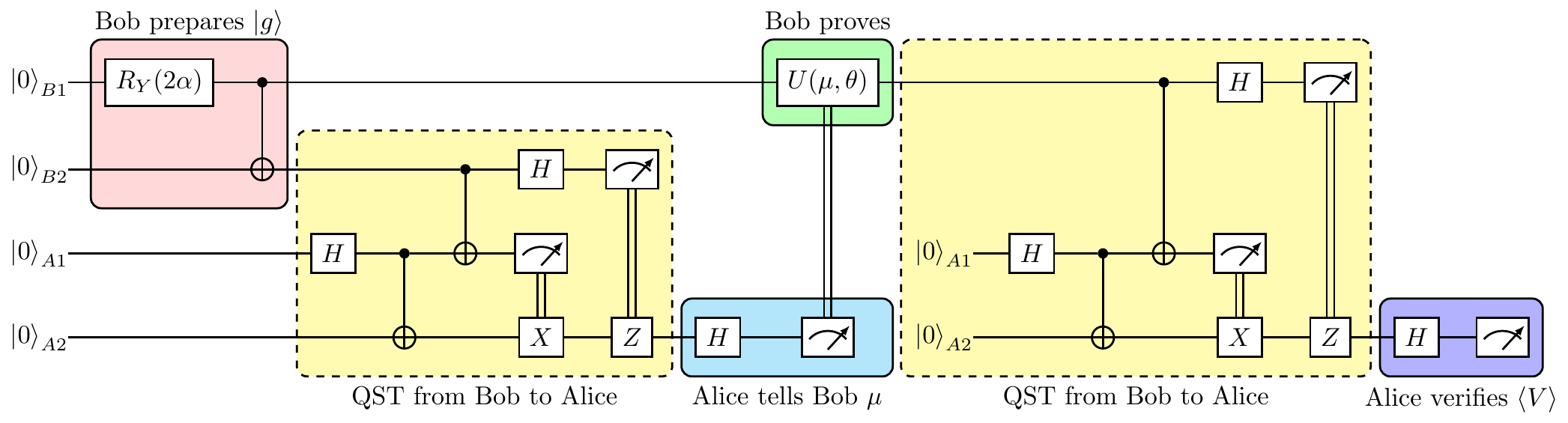}
    \caption{The full quantum circuit for quantum interactive proofs using the minimal model of quantum energy teleportation. Circuit complexity is 13 and circuit width is 4. In the circuit, each $H$ is an Hadamard gate.}
    \label{fig:circuit}
\end{figure*}

\section{\label{sec:minimal}QIP using minimal QET}

\subsection{Quantum Circuit of QIP}
Here we establish a QIP using the simplest model of QET. For this, we use the following Hamiltonian
\begin{align}
\label{eq:ham_min}
    H&=H_1+H_2+V,\\
    H_n&=hZ_n+\frac{h^2}{\sqrt{h^2+k^2}},~(n=1,2)\\
    V&=2kX_1X_2+\frac{2k^2}{\sqrt{h^2+k^2}},
\end{align}
where $k$ and $h$ are positive real values and $n=1,2$ are prover's qubits. The constant terms are added so that the ground state $\ket{g}$ of $H$ returns the zero mean energy for all local and global Hamiltonians: 
\begin{equation}
\label{eq:ground}
\bra{g}H\ket{g}=\bra{g}H_1\ket{g}=\bra{g}H_2\ket{g}=\bra{g}V\ket{g}=0. 
\end{equation}
This can be explicitly checked using the ground state of $H$
\begin{equation}
\label{eq:groundstate}
    \ket{g}=\frac{1}{\sqrt{2}}\sqrt{1-\frac{h}{\sqrt{h^2+k^2}}}\ket{00}-\frac{1}{\sqrt{2}}\sqrt{1+\frac{h}{\sqrt{h^2+k^2}}}\ket{11}.
\end{equation}
Although $\ket{g}$ is the ground state of the total Hamiltonian $H$, it should be noted that it is neither a ground state nor an eigenstate of $H_n,V, H_n+V~(n=1,2)$. The essence of QET is to extract negative ground state energy of those local and semi-local Hamiltonians. Before the prover and verifier start QIP, they need to make sure that \eqref{eq:groundstate} is manifested and the prover must not change the parameters.  

Moreover, according to the formula~\eqref{eq:theta}, the angle $\theta$ is constrained by 
\begin{equation}
\label{eq:theta_minimal}
\begin{split}
\cos 2 \theta = \frac{h^2 + 2k^2}{\sqrt{(h^2 + 2 k^2)^2 + (hk)^2}}, \\
\sin 2 \theta= \frac{hk}{\sqrt{(h^2 + 2 k^2)^2 + (hk)^2}}.
\end{split}
\end{equation}

We present the full quantum circuit of QIP using the minimal QET model in Fig.~\ref{fig:circuit}. The idea of using QST for long-distance QET on a large-scale quantum network was first proposed by one of the authors~\cite{2023arXiv230111884I}.  

\begin{tcolorbox}
[colframe = blue!50!black, 
    colback = white, 
    title = Objective of the protocol]
Prover's goal is to encode a pair of positive real numbers $(k,h)$ into the quantum state $\ket{g(k,h)}$ and prove that it is correctly generated.
\end{tcolorbox}

First, the prover (Bob) chooses parameters $k$ and $h$ to create the ground state $\ket{g(k,h)}$. Then using QST the prover transfers the first qubit (witness) to the verifier (Alice) who observes the qubit in the $X$ basis and announces the result to the prover by a classical channel. $\mu=-1,+1$ appears with the equal probability, hence Bob cannot predict in advance which $\mu$ will appear. The prover performs a conditional operation on the second qubit according to the result of the verifier and sends it to the verifier by QST. The verifier computes $V+H_2$ using the $X$-basis and the $Z$-basis. (In Fig.~\ref{fig:circuit}, the verifier can compute $V$ by measuring $X_2$ in the $X$-basis.)

In the rest of this article, we confirm the following properties:
\begin{itemize}
\item (Zero-knowledge) Even if the verifier observes the state sent by the prover, no information about the ground state of the prover can be obtained.
\item (Completeness) If $\ket{\psi}$ is the ground state of the Hamiltonian~\eqref{eq:ham_min}, then there is a $\theta$ such that $\langle E_B\rangle<0$.
\item (Soundness) If $\ket{\psi}$ is not the ground state of the Hamiltonian~\eqref{eq:ham_min}, then the verifier observes positive energy no matter what state the prover uses, except for some small probability.
\end{itemize}

Once all verifiers have computed their local energies, either centralized or decentralized~\cite{ben2018scalable,9152704,10.1145/3133956.3134104,Kiktenko_2018}, a consensus can be aggregated on whether all verifiers accept the state of the prover.

\subsection{Zero-knowledge}
The purpose of this section is to show the process above is zero-knowledge proof in the sense that the verifier cannot obtain any knowledge of $\ket{g(k,h)}$, although one qubit of it is given. The inputs, variables and output results used in this protocol are $(k,h)$ to generate the ground state, Alice's projective measurement results $\mu$, $U_B(\mu,\theta)$ that Bob uses for his conditional operations, and the energy that Alice observes. Of these, only $(k,h)$ should be kept secret by Bob, and any other values disclosed would not lead to information about Bob's ground state~$\ket{g(k,h)}$.

Using the density matrix eq.~\eqref{eq:rho_QET}, we evaluate the expectation value $\langle H_1 \rangle=\Tr[\rho_\text{QET} H_1]$ and $\langle V \rangle=\Tr[\rho_\text{QET}V]$. The explicit formulas of them are obtained analitycally as follows: 
\begin{equation}
\label{eq:HandV}
\begin{split}
    \langle H_1\rangle &= \frac{h}{2\sqrt{h^2+k^2}} \Big[ 2k \sin 2 \theta + h (1 - \cos 2 \theta) \Big]\\
    \langle V \rangle &=\frac{2k}{\sqrt{h^2+k^2}} \Big[ -h\sin 2 \theta+ k (1 - \cos 2 \theta) \Big].
\end{split}
\end{equation}
By combing them, we find that $\langle E_B \rangle=\Tr[\rho_\text{QET}(H_1+V)]$ is 
\begin{equation}
\label{eq:Bob_energy}
\begin{split}
\langle E_B  \rangle =& \frac{1}{ \sqrt{h^2 + k^2}} \Big[ - k h \sin 2\theta  \\ &+ (h^2 + 2k^2 ) ( 1- \cos 2 \theta) \Big],
\end{split}
\end{equation}
which is equal to the analytical formula \eqref{eq:QET}.

We first confirm that converting from $(k,h)$ to $\theta$ is a one-way function, since there are uncountably many combinations of $k$ and $h$ taking the same $\sin\theta$. 

\begin{proposition}
Knowing $\theta$ only does not yield any knowledge of the ground state $\ket{g(k,h)}$.
\end{proposition}
\begin{proof}
This statement follows from the following fact: for any given $\theta\in[0,2\pi)$, the set 
\begin{equation}
\left\{(k,h):\sin\theta=\frac{hk}{\sqrt{(h^2+2k^2)^2+(kh)^2}},k>0,h>0\right\},    
\end{equation}
 is an uncountable set. 
\end{proof}

This means that Bob may disclose information about $\theta$. Our protocol does not use this property, but it would be beneficial to apply this study to another case.

For our protocol to be executed with zero knowledge, the following statement is critical.
\begin{theorem}
Knowing the measurement value of $X_1X_2$ and $Z_1$ does not yield any knowledge of the ground state $\ket{g(k,h)}$.
\end{theorem}
\begin{proof}
This can be explained as follows. 
Let $\rho(k,h,\theta)=\rho_{QET}$~\eqref{eq:rho_QET} be the density matrix obtained after operations by Alice and Bob. For any given $\theta\in[0,2\pi)$ and any given $V,Z\in\mathbb{R}$, the sets 
\begin{align}
&\left\{(k,h):V=\Tr[\rho(k,h,\theta) X_1X_2],k>0,h>0\right\}\\
&\left\{(k,h):Z=\Tr[\rho(k,h,\theta) Z_1],k>0,h>0\right\}
\end{align}
 are uncountable sets, if they are not empty. The right hand sides of these conditions are given in eq.~\eqref{eq:HandV}. Therefore the verifier can know those values statistically, but it is impossible to know $(k,h)$ with those values.
\end{proof}

Since the set of combinations of $(k,h)$ that output the same value as the teleported energy $\langle E_B\rangle$~\eqref{eq:Bob_energy} is also an uncountable set, the map from $(k,h)$ to the total energy that Alice observes is also a one-way function.

Admitting that Alice cannot get any knowldge about Bob's input by her projective measurement, we can summarize our statements as follows.
\begin{theorem}
    As long as the input values are kept secret, no information about the ground state $\ket{g(k,h)}$ can be obtained by any observation through the QET protocol.
\end{theorem}

\begin{figure*}
    \centering
    \includegraphics[width=0.49\linewidth]{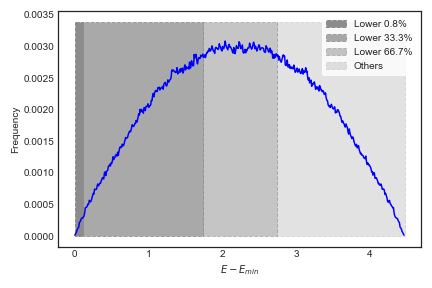}
    \includegraphics[width=0.49\linewidth]{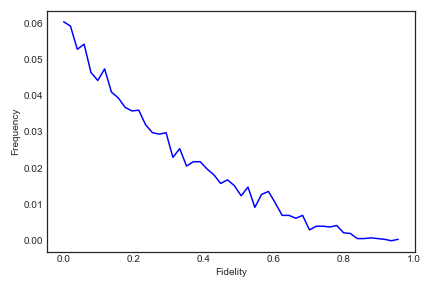}
    \caption{[Left] Histogram of the energy computed by the 500 random unitary operators that generate initial states. The energy was computed for the Hamiltonian $(k=h=1)$ by 600 different $\theta$s sampled between 0 and $2\pi$ at equal intervals. $E_\text{min}$ corresponds to the lowest energy with the exact ground state at $k=h=1$. Only the lower 0.8\% of the total samples generate negative energy value. [Right] Histogram of the fidelity between the exact ground state and the random initial states.}
    \label{fig:histgram}
\end{figure*}
\subsection{Analysis on Completeness and Soundness}
To achieve a QIP, we need to make sure that the state $\ket{g}$ is accepted with a high probability if the correct process is executed and rejected with a high probability otherwise. Let us assume that quantum communication via QST is completely secure and there are no errors. 

\subsubsection{\textbf{Completeness}}
The proof of completeness is straightforward, using the analytical formula of the teleported energy~\eqref{eq:QET}, which is generally negative as long as Bob generates the correct ground state and uses the correct parameters $\mu,\theta$. In Fig. \ref{eminF}, we plot the sensitivity of the teleported energy with respect to $\theta$~determined by eq.\eqref{eq:theta_minimal} (or equivalently by eq.~\eqref{eq:theta}). This is the case where he uses the correct ground state but puts a slightly different $\theta$ in his conditional operation~\eqref{eq:operation}. In the figure, $\delta$ is a difference from the correct $\theta$. It is clear that $\theta$ can not be chosen randomly, as there is functional dependence on $\delta$ and therefore $\theta$. However, the result shows that the teleported energy is robust to minute fluctuations in $\theta$ and does not significantly affect the verification of the proof. For example, Alice should accept Bob's input only if she observes negative energy. 

\begin{figure}[H]
\centering
    \includegraphics[scale=0.6]{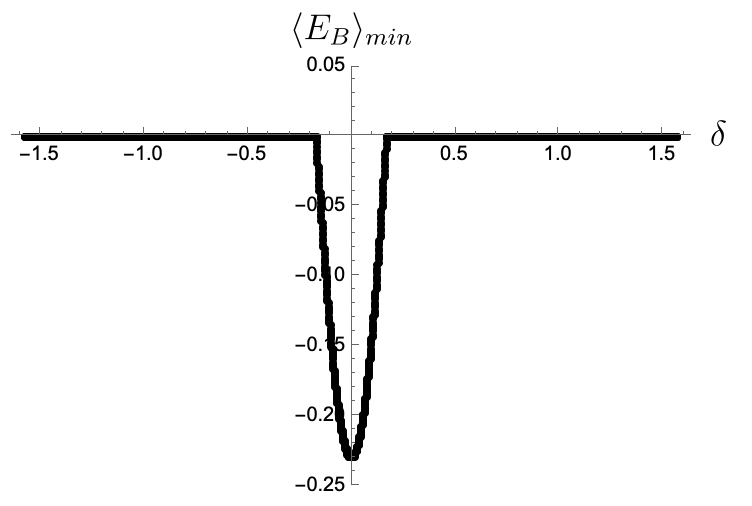}
    \caption{This figure shows the minimum value for $\langle E_B \rangle$ as a function of $\delta$. When $\delta=0$, it corresponds to the correct $\theta$.}
    \label{eminF}
\end{figure}

\subsubsection{\textbf{Soundness}}
Finally we check soundness. We show Alice observes negative energy if Bob sends her the correct ground state of a Hamiltonian and he uses the correct $\theta$. For this let us investigate how successful an attack on the protocol would be if Bob cannot generate the correct ground state~\eqref{eq:groundstate}. In order to confirm soundness, random initial states will be used in order to check that the average energy is only negative for the correct ground state. The random initial state will be of the form
\begin{equation}
\label{eq:random}
    \ket{\psi} = U\ket{00},
\end{equation}
where $U$ is a unitary matrix.  Using this, the random density matrix after conditional operations become
\begin{equation}
    \rho(\theta,\mu) = U_B(\theta) P_A(\mu) \ket{\psi}\bra{\psi} P_A (\mu) U_{B}^{\dagger}(\theta),
\end{equation}
where it is unknown what constraints are placed on $\theta$. To compute the average for Bob's energy, the result is
\begin{equation}
\label{eq:random_energy}
     \langle E_B (h,k,\theta,\mu) \rangle = \tr \Big[ \rho(\theta,\mu) (H_1(h,k) + V(h,k)) \Big].
\end{equation}

Moreover to evaluate the distance between the exact ground state, we consider the fidelity defined by
\begin{equation}
\label{eq:fidelity}
    F(h,k) = | \braket{g(h,k) | \psi} |^2.
\end{equation}
For a given choice of $h,k$, the fidelity can be plotted against Bob's energy for the same choice of $h,k$. This gives a visualization of how close the random initial states are to the correct ground state.

To evaluate Bob's attack, we generated 500 random unitary matrices $U$~\eqref{eq:random} according to the Haar measure. For $\theta$ in Bob's operation $U_B(\mu,\theta)$, 600 points were sampled at equal intervals from 0 to $2\pi$. Therefore it contains data corresponding to 300,000 random attacks on the protocol. We are assuming that Bob does not have the ground state~\eqref{eq:groundstate}, so he does not know about the correct $\theta$~\eqref{eq:theta} either. In the left panel of  Fig.~\ref{fig:histgram}, we show the histogram of the energy $E-E_\text{min}$ that Alice observes. Here $E$ is as given by eq.\eqref{eq:random_energy} and $E_\text{min}$ is the smallest negative energy ~\eqref{eq:QET} that Alice can see if Bob uses the correct ground state. Of the 300000 data sampled, only 0.8\% had negative $\langle E_B (h,k,\theta,\mu) \rangle$, which is indicated by the dark-colored area in the figure. Therefore Alice can assume that Bob has generated the appropriate ground state with a very high probability.

Furthermore, in the right panel of Fig.~\ref{fig:histgram}, we plot the histogram of the fidelity~\eqref{eq:fidelity} between the random states~\eqref{eq:random} and the exact ground state~\eqref{eq:groundstate}. We can also see that, although it is only a two-qubit state,  the possibility of approximating the proper ground state with a random state is close to zero. 

In summary, If Bob can generate the appropriate ground state, Alice can observe negative energy with high probability, so Bob's proof is accepted with high probability (Completeness). Additionally, if Bob is unable to generate a proper ground state, Alice will observe positive energy with a sufficiently large probability, so the chance that the proof will be correctly accepted even though Bob does not have a ground state is very small (Soundness). 

\begin{figure*}
    \centering
    \includegraphics[width=0.49\linewidth]{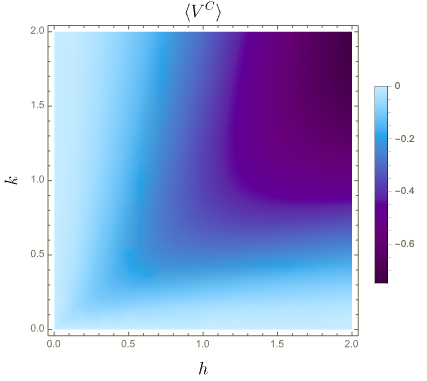}
\includegraphics[width=0.49\linewidth]{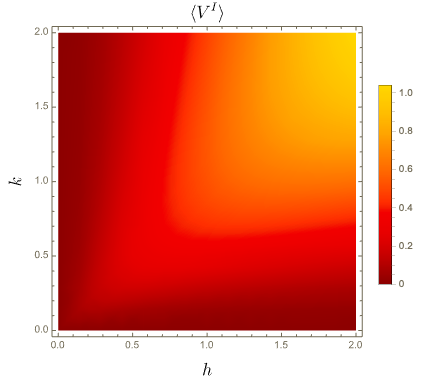}
\includegraphics[width=0.49\linewidth]{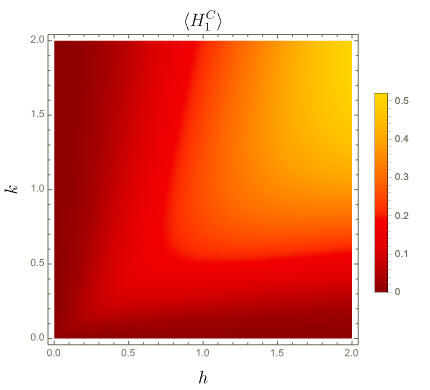}
\includegraphics[width=0.49\linewidth]{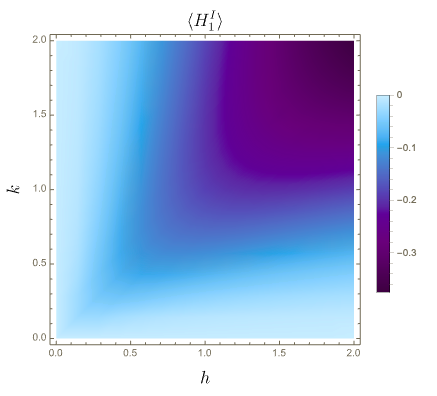}
    \caption{Expectation values of the interaction terms $\langle V^{C}\rangle,\langle V^{I}\rangle$
    and the single qubit operators $\langle H_{1}^{C} \rangle,\langle H_{1}^{I} \rangle$. See eq.~\eqref{eq:expectation_val} for the definition. $\langle V^{C}\rangle$ and $\langle H^{I}\rangle$ are always negative, whereas $\langle V^{I}\rangle$ and $\langle H^{C}\rangle$ are always positive, by which the prover can distinguish whether $Q_1$ or $Q_2$ was applied.}
    \label{QET_Lie_True}
\end{figure*}
\begin{table*}
    \centering
    \begin{tabular}{|c|c|c|c|c|}
\cline{1-5}
& \multicolumn{1}{c}{} & \multicolumn{1}{c}{$\langle H_{1}^{C} \rangle = \sum_{\mu,\nu} \langle H_1(\mu=\nu)\rangle$}  &   \multicolumn{1}{c}{}&  \\ \hline
Shot Count & $(h,k)=(1,0.2)$&$(h,k)=(1,0.5)$& $(h,k)=(1,1)$ &$(h,k)=(1.5,1)$ \\\cline{1-5}
$n_\text{shot}=100$ & $0.0204$ & $0.1144$& $0.3072$ &$0.4380$ \\\cline{1-5}
$n_\text{shot}=1000$ & $0.0624$ & $0.1564$& $0.2152$ &$0.3600$ \\\cline{1-5}
$n_\text{shot}=10000$ & $0.0572$ & $0.1870$& $0.2422$ &$0.3666$ \\\cline{1-5}
Analytical value& 0.0521 &0.1873  &0.2598 & 0.3480 \\\hline
& \multicolumn{1}{c}{} & \multicolumn{1}{c}{$\langle H_{1}^{I}\rangle = \sum_{\mu,\nu} \langle H_1(\mu \neq \nu)\rangle$}  &   \multicolumn{1}{c}{}&  \\ \hline
Shot Count & $(h,k)=(1,0.2)$&$(h,k)=(1,0.5)$& $(h,k)=(1,1)$ &$(h,k)=(1.5,1)$ \\\cline{1-5}
$n_\text{shot}=100$ & $-0.0194$ & $-0.0856$& $-0.2328$ &$-0.2218$ \\\cline{1-5}
$n_\text{shot}=1000$ & $-0.0194$ & $-0.0836$& $-0.1968$ &$-0.2070$ \\\cline{1-5}
$n_\text{shot}=10000$ & $-0.0192$ & $-0.0948$& $-0.1776$ &$-0.2044$ \\\cline{1-5}
Analytical value& -0.0193 &-0.0955 &-0.1873& -0.2058 \\\hline
& \multicolumn{1}{c}{} & \multicolumn{1}{c}{$\langle V^{C} \rangle = \sum_{\mu,\nu} \langle V(\mu=\nu)\rangle$}  &   \multicolumn{1}{c}{}&  \\ \hline
Shot Count & $(h,k)=(1,0.2)$&$(h,k)=(1,0.5)$& $(h,k)=(1,1)$ &$(h,k)=(1.5,1)$ \\\cline{1-5}
$n_\text{shot}=100$ & $-0.0816$ & $-0.2728$& $ -0.4656$ &$-0.4436$ \\\cline{1-5}
$n_\text{shot}=1000$ & $-0.0712$ & $-0.2546$& $ -0.3618$ &$-0.5064$ \\\cline{1-5}
$n_\text{shot}=10000$ & $-0.0716$ & $-0.2644$& $-0.3802$ &$-0.5004$ \\\cline{1-5}
Analytical value& -0.0701 &-0.2599 & -0.3746 & -0.4906 \\\hline
& \multicolumn{1}{c}{} & \multicolumn{1}{c}{$\langle V^{I} \rangle = \sum_{\mu,\nu}\langle V(\mu \neq \nu)\rangle$}  &   \multicolumn{1}{c}{}&  \\ \hline
Shot Count & $(h,k)=(1,0.2)$&$(h,k)=(1,0.5)$& $(h,k)=(1,1)$ &$(h,k)=(1.5,1)$ \\\cline{1-5}
$n_\text{shot}=100$ & $-0.0176$ & $0.2672$& $0.4142$ &$0.5494$ \\\cline{1-5}
$n_\text{shot}=1000$ & $0.0816$ & $0.3132$& $0.3422$ &$0.6294$ \\\cline{1-5}
$n_\text{shot}=10000$ & $0.0766$ & $0.2954$& $0.5250$ &$0.5942$ \\\cline{1-5}
Analytical value& 0.0727 &0.3058 &0.5198 & 0.6171 \\\hline
\end{tabular}
\caption{\label{tab:result}This is a comparison between the number of shot counts in the protocol for computing expectation values. For the protocol to work reliably, the number of shots should be greater than 1000. The errors $\sigma$ correspond to the values decrease with increasing the number of iterations $n_\text{shot}$. For example, for $n=100$, $\sigma \sim 0.1$, and for $n_\text{shot}=1000$, $\sigma \sim 0.01$.}
\end{table*}

\section{Applications: Quantum State Distinguishability}
It is widely known that the Quantum Circuit Distinguishability (QCD) is a QIP-complete problem~\cite{1443098} and that the Quantum State Distinguishability (QSD) is a QSZK (Quantum Statistical Zero-Knowledge)-complete problem~\cite{1181970}. Based on QET, now let us consider a simple QSD problem as follows. The corresponding quantum circuit can be easily drawn as presented in Fig.~\ref{fig:circuit}.

\begin{algorithm}
\caption{QSD by Quantum Energy Teleportation}\label{alg:cap}
\hspace*{\algorithmicindent} \textbf{Input} The ground state $\ket{g}$ of a Hamiltonian~\eqref{eq:condition}\\
    \hspace*{\algorithmicindent} \textbf{Output} $\rho_1$ or $\rho_2$ (equivalently $Q_1$ or $Q_2$)
\begin{algorithmic}
\Require $n_\text{shot}>0$
\For{$n <n_\text{shot}$}
\State Prover measures $X_1$, gets $\mu\in\{-1,1\}$ and tells $\mu$ to the verifier. 

\State Verifier applies $Q_1=U(+\mu)$ or $Q_2=U(-\mu)$ randomly to the sate
\State Prover measures $X_2$ and $Z_2$
\EndFor 
\If{$\langle E_{B}\rangle<0$}
\State Prover finds $Q_1$ was applied.  
\Else
\State Prover finds $Q_2$ was applied.  
\EndIf
\end{algorithmic}
\end{algorithm}

\begin{enumerate}
    \item Prover creates the ground state and performs the projective measurement on the first qubit. 
    \item Prover sends the second qubit, the angle $\theta$ and the measurement value $\mu$ to the verifier.  
    \item Verifier selects $Q_1=U_B(\mu)$ or $Q_2=U_B(-\mu)$ at random, applies it to the qubit and sends it back to the prover.  
    \item Prover measures the local energy and tells the verifier whether the verifier applied $Q_1$ or $Q_2$.
    \item Verifier accepts if the answer is correct and rejects otherwise. 
\end{enumerate}
In step 3) of the process, the verifier randomly decides whether to use (faithful) or not use (non-faithful) the same value of $\mu$ as that conveyed by the prover. This choice of faithful or non-faithful shall be made only once at the beginning and cannot be changed thereafter. Thus, if the verifier is faithful to the prover every time, state $\rho_1$ is obtained, while if the verifier is not faithful every time, state $\rho_2$ is obtained:
\begin{align}
\begin{aligned}
    \rho_1&=\sum_{\mu\in\{-1,1\}}U_B(\mu)P_A(\mu)\ket{g}\bra{g}P_A(\mu) U_B(\mu)\\
    \rho_2&=\sum_{\mu\in\{-1,1\}}U_B(-\mu)P_A(\mu)\ket{g}\bra{g}P_A(\mu)U_B(-\mu)
\end{aligned}
\end{align}
The prover's goal is to distinguish the difference between $\rho_1$ and $\rho_2$ using QET. In the following, we explain that in the case of the minimal model, the prover can distinguish the states quite accurately. For the smallest model, the certifier can do this by measuring only either $Z_2$ or $X_1X_2$.

Let $\mu$ be the prover's measurement result and $\nu$ be the verifier's value for the operation
\begin{equation}
U_{B} (\nu,\theta) = \cos \theta I - i \nu \sin \theta \sigma_{B}.
\end{equation}
If $\mu=\nu$, it is equal to the operator discussed before~\eqref{eq:operation}.

After the prover's post-measurement, the result is
\begin{equation}
\rho' (\mu) = \frac{1}{p_A(\mu)}P_A (\mu) \ket{g} \bra{g} P_A (\mu)
\end{equation}
and after the verifier's operation on this state, the density matrix becomes
\begin{equation}
\rho (\mu, \nu) = U_B (\nu) \rho' (\mu) U_{B}^{\dagger} (\nu).
\end{equation}

We evaluate the expectation value $\langle H_1 (\mu,\nu) \rangle=\Tr[\rho (\mu, \nu) H_1]$ and $\langle V(\mu,\nu) \rangle=\Tr[\rho (\mu, \nu) V]$. The average of $H_1$ is found to be
\begin{equation}
\begin{split}
    \langle H_1 (\mu,\nu) \rangle = \frac{h}{4\sqrt{h^2+k^2}} \Big[ 2k \mu \nu \sin 2 \theta+ h (1 - \cos 2 \theta) \Big].
\end{split}
\label{Hav}
\end{equation}
When $\mu = \nu$ as in the standard QET protocol, the correct results are obtained (with a factor of 1/2 as the summation is relaxed), and $\langle H_1(\mu,\nu) \rangle$ is always non-negative. On the other hand,  $\langle H_1(\mu,\nu) \rangle$ can be negative when $\mu\neq\nu$. For the average of $V$, the result is
\begin{equation}
\begin{split}
    \langle V (\mu,\nu) \rangle =& \frac{k}{\sqrt{h^2+k^2}} \Big[ -h \mu \nu \sin 2 \theta + k (1 - \cos 2 \theta) \Big].
\end{split}
\label{Vav}
\end{equation}
$\langle V (\mu,\nu) \rangle$ is always non-positive when $\mu=\nu$, but is always non-negative when $\mu\neq\nu$. Combing these two formulas, we find that $\langle E_B (\mu,\nu) \rangle=\Tr[\rho(\mu,\nu)(H_1+V)]$ is 
\begin{equation}
\begin{split}
\langle E_B (\mu, \nu) \rangle =& \frac{1}{2 \sqrt{h^2 + k^2}} \Big[ - k h \mu \nu \sin 2\theta  \\ &+ (h^2 + 2k^2 ) ( 1- \cos 2 \theta) \Big].
\end{split}
\end{equation}
One can confirm that $\langle E_B (\mu, \nu) \rangle$ is negative in general for $\mu=\nu$ and positive in general for $\mu\neq\nu$.

For an operator $O\in\{V,H_1\}$, we define the expectation value by
\begin{align}
\label{eq:expectation_val}
\begin{aligned}
    \langle O^C\rangle &=\sum_{\mu\in\{-1,1\}}\langle O(\mu,\mu)\rangle\\
    \langle O^I\rangle &=\sum_{\mu\in\{-1,1\}}\langle O(\mu,-\mu)\rangle
\end{aligned}
\end{align}
and plot in Fig.~\ref{QET_Lie_True}. One can find that $\langle V^{C}\rangle$ and $\langle H_1^{I}\rangle$ are always negative, whereas $\langle V^{I}\rangle$ and $\langle H_1^{C}\rangle$ are always positive. Therefore, this gives a direct operational protocol for Bob to determine whether $\rho_1$ or $\rho_2$ was created by Alice's operation $Q_1$ or $Q_2$.

Here we discuss the number of repetitions of operations required to execute the protocol. For this, we used \textit{qasm\_simulator} available in Qiskit package to simulate the protocol with the circuit. Our result is summarized in Table~\ref{tab:result}. The $\mu=\nu$ case is essentially the same as given in the previous work by one of the authors~\cite{Ikeda:2023uni}. It can be seen from Table that with an ideal quantum circuit, it takes at most 1000 calculations (approximately $O(500)\mu s$) for Bob to distinguish the states.

\section{\label{sec:QMIP}Quantum Multi-Prover Interactive Proofs and Future Direction}
The work presented in this paper offers tangible practical benefits for society in the quantum technological age. Utilizing a modified version of QET as a protocol to determine verification and trust across a network reveals a technique that can increase security, whilst also decreasing the need for current cryptographic techniques to be adapted for quantum hardware. 

Extending our protocol to distributed QIPs where multiple ($N$) verifiers are distributed on a quantum network is straightforward. It is also straightforward to extend our protocol to Quantum Multi-Prover Interactive Proofs with entangled provers (QMIP$^*$), by introducing many energy suppliers into the system. Technically, provers delegate their operation for supplying energy to a single verifier (Alice), who performs projective measurements on the qubits teleported by the provers (See Steps 1 to 3 in Fig.~\ref{fig:QMIP}. This is indeed a QMIP$^*$ protocol since the ground state of a many-body quantum system is entangled in general, and provers do not communicate after Alice's post measurements. If each prover is far enough away from the other, their local Hamiltonians have no overlap with each other, so each can independently perform the proof. Since NEXPTIME$\subsetneq$QIMP$^*$, our protocol is ultimately secure even though provers have the best quantum computational resources. Since all proofs can be performed in parallel computation, the total execution time is the same as for a single prover. Remember that the authentication takes only $O(500)~\mu s$ (evaluated by the CNOT duration.) even if the provers and the verifier communicate many times. Zero-knowledge and Completeness can be shown as done in the previous section. We leave the rigorous proof of Soundness as an open question for subsequent studies.

\begin{tcolorbox}
[colframe = blue!50!black, 
    colback = white, 
    title = Quantum Multi-Prover Interactive Proofs with entangled provers (QMIP$^*$) using QET and QST]
\begin{itemize}
    \item[1.] All provers agree on a Hamiltonian and generate the entangled ground state $\ket{g}$.
    \item[2.] Each of them sends Alice a single qubit from their local subsystem by QST.
    \item[3.] Alice performs projective measurements to each of the qubits sent and announces her results to them.
    \item[4.] Each verifier performs a conditional operation to another qubit in their local subsystem and sends it to Alice by QST.
    \item[5.] Alice calculates the local energy and decides whether to accept each of them and the ground state $\ket{g}$ based on that energy.
\end{itemize}
\end{tcolorbox}
\begin{figure}[H]
    \centering
    \includegraphics[width=\linewidth]{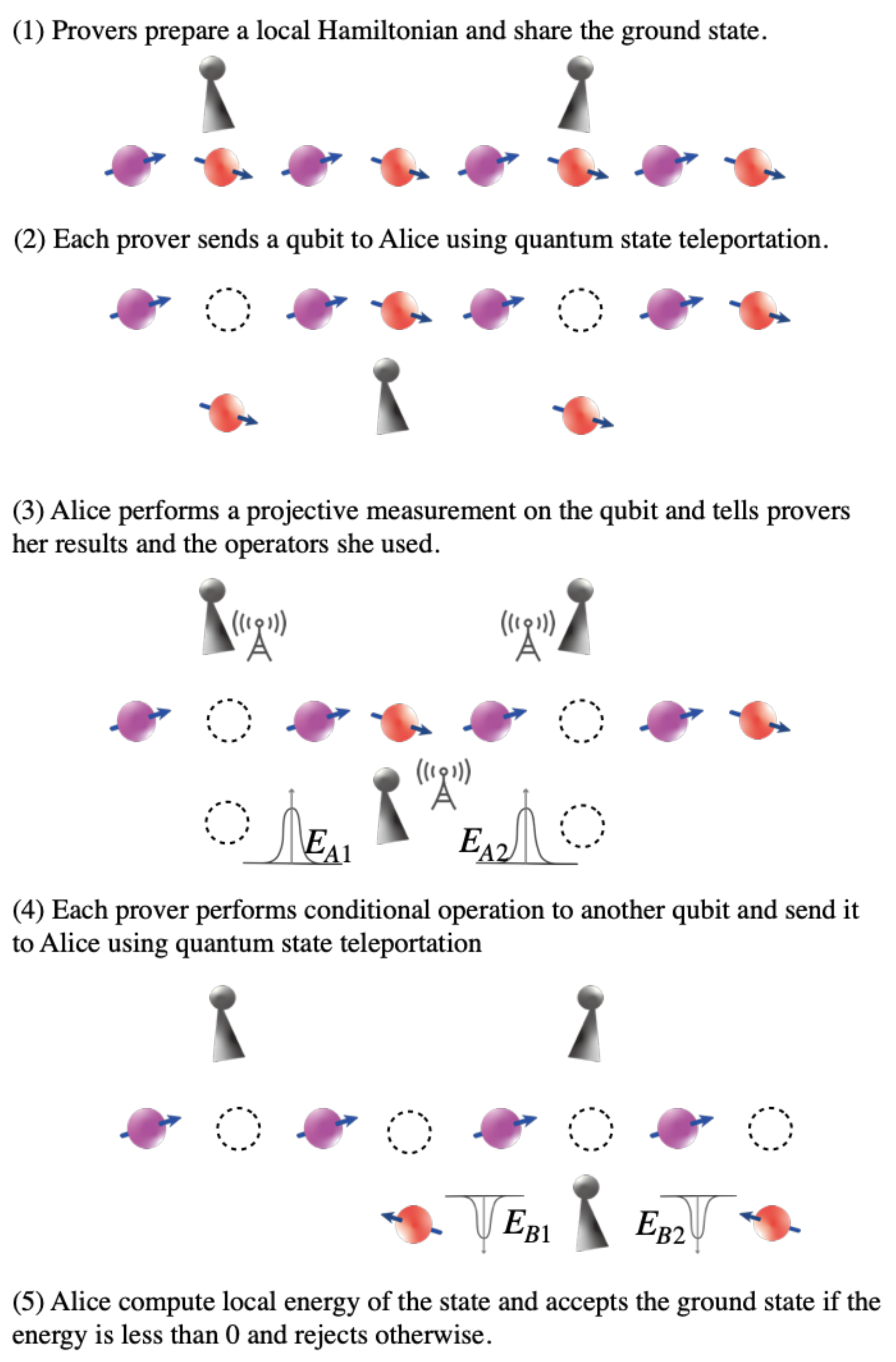}
    \caption{Schematic picture of Quantum Multi-Prover Interactive Proof systems.}
    \label{fig:QMIP}
\end{figure}

\section{Conclusion}
In summary, this paper presents a novel QIP, which utilizes the properties of QET and QST such that information transmitted across a quantum network can be verified without the receiving party having full knowledge of the other node's system. We showed analytically how this protocol works in the minimal scenario, by verifying the soundness, completeness, and its zero-knowledge properties. Importantly, the general protocol can be extended up to $N$-parties as previously discussed. This work constitutes a zero-knowledge proof, which will have significant implications for future quantum network design, cryptography and secure quantum communication.

\section*{Acknowledgement}
The work of KI was supported by the U.S. Department of Energy, Office of Science, National Quantum Information Science Research Centers, Co-design Center for Quantum Advantage (C2QA) under Contract No.DESC0012704. We acknowledge the use of IBM quantum simulators and Qiskit. 

\bibliographystyle{IEEEtran}
\bibliography{ref}

\begin{thebibliography}{10}
\providecommand{\url}[1]{#1}
\csname url@samestyle\endcsname
\providecommand{\newblock}{\relax}
\providecommand{\bibinfo}[2]{#2}
\providecommand{\BIBentrySTDinterwordspacing}{\spaceskip=0pt\relax}
\providecommand{\BIBentryALTinterwordstretchfactor}{4}
\providecommand{\BIBentryALTinterwordspacing}{\spaceskip=\fontdimen2\font plus
\BIBentryALTinterwordstretchfactor\fontdimen3\font minus
  \fontdimen4\font\relax}
\providecommand{\BIBforeignlanguage}[2]{{%
\expandafter\ifx\csname l@#1\endcsname\relax
\typeout{** WARNING: IEEEtran.bst: No hyphenation pattern has been}%
\typeout{** loaded for the language `#1'. Using the pattern for}%
\typeout{** the default language instead.}%
\else
\language=\csname l@#1\endcsname
\fi
#2}}
\providecommand{\BIBdecl}{\relax}
\BIBdecl

\bibitem{PhysRevLett.70.1895}
\BIBentryALTinterwordspacing
C.~H. Bennett, G.~Brassard, C.~Cr\'epeau, R.~Jozsa, A.~Peres, and W.~K.
  Wootters, ``Teleporting an unknown quantum state via dual classical and
  einstein-podolsky-rosen channels,'' \emph{Phys. Rev. Lett.}, vol.~70, pp.
  1895--1899, Mar 1993. [Online]. Available:
  \url{https://link.aps.org/doi/10.1103/PhysRevLett.70.1895}
\BIBentrySTDinterwordspacing

\bibitem{furusawa1998unconditional}
A.~Furusawa, J.~L. S{\o}rensen, S.~L. Braunstein, C.~A. Fuchs, H.~J. Kimble,
  and E.~S. Polzik, ``Unconditional quantum teleportation,'' \emph{science},
  vol. 282, no. 5389, pp. 706--709, 1998.

\bibitem{2015NaPho...9..641P}
S.~{Pirandola}, J.~{Eisert}, C.~{Weedbrook}, A.~{Furusawa}, and S.~L.
  {Braunstein}, ``{Advances in quantum teleportation},'' \emph{Nature
  Photonics}, vol.~9, no.~10, pp. 641--652, Oct. 2015.

\bibitem{takeda2013deterministic}
S.~Takeda, T.~Mizuta, M.~Fuwa, P.~Van~Loock, and A.~Furusawa, ``Deterministic
  quantum teleportation of photonic quantum bits by a hybrid technique,''
  \emph{Nature}, vol. 500, no. 7462, pp. 315--318, 2013.

\bibitem{HOTTA20085671}
\BIBentryALTinterwordspacing
M.~Hotta, ``A protocol for quantum energy distribution,'' \emph{Physics Letters
  A}, vol. 372, no.~35, pp. 5671--5676, 2008. [Online]. Available:
  \url{https://www.sciencedirect.com/science/article/pii/S0375960108010347}
\BIBentrySTDinterwordspacing

\bibitem{2009JPSJ...78c4001H}
M.~{Hotta}, ``{Quantum Energy Teleportation in Spin Chain Systems},''
  \emph{Journal of the Physical Society of Japan}, vol.~78, no.~3, p. 034001,
  Mar. 2009.

\bibitem{2015JPhA...48q5302T}
J.~{Trevison} and M.~{Hotta}, ``{Quantum energy teleportation across a
  three-spin Ising chain in a Gibbs state},'' \emph{Journal of Physics A
  Mathematical General}, vol.~48, no.~17, p. 175302, May 2015.

\bibitem{2023arXiv230111884I}
K.~{Ikeda}, ``{Long-range quantum energy teleportation and distribution on a
  hyperbolic quantum network},'' \emph{arXiv e-prints}, p. arXiv:2301.11884,
  Jan. 2023.

\bibitem{Ikeda:2023uni}
K.~Ikeda, ``{Realization of Quantum Energy Teleportation on Superconducting
  Quantum Hardware},'' 1 2023.

\bibitem{Ikeda:2023xmf}
------, ``{Investigating global and topological order of states by local
  measurement and classical communication: Study on SPT phase diagrams by
  quantum energy teleportation},'' 2 2023.

\bibitem{PhysRevD.107.L071502}
\BIBentryALTinterwordspacing
------, ``Criticality of quantum energy teleportation at phase transition
  points in quantum field theory,'' \emph{Phys. Rev. D}, vol. 107, p. L071502,
  Apr 2023. [Online]. Available:
  \url{https://link.aps.org/doi/10.1103/PhysRevD.107.L071502}
\BIBentrySTDinterwordspacing

\bibitem{10.1007/978-3-540-30538-5_31}
J.~Kempe, A.~Kitaev, and O.~Regev, ``The complexity of the local hamiltonian
  problem,'' in \emph{FSTTCS 2004: Foundations of Software Technology and
  Theoretical Computer Science}, K.~Lodaya and M.~Mahajan, Eds.\hskip 1em plus
  0.5em minus 0.4em\relax Berlin, Heidelberg: Springer Berlin Heidelberg, 2005,
  pp. 372--383.

\bibitem{broadbent2022qma}
A.~Broadbent and A.~B. Grilo, ``Qma-hardness of consistency of local density
  matrices with applications to quantum zero-knowledge,'' \emph{SIAM Journal on
  Computing}, vol.~51, no.~4, pp. 1400--1450, 2022.

\bibitem{Aharonov2009}
\BIBentryALTinterwordspacing
D.~Aharonov, D.~Gottesman, S.~Irani, and J.~Kempe, ``The power of quantum
  systems on a line,'' \emph{Communications in Mathematical Physics}, vol. 287,
  no.~1, pp. 41--65, Apr 2009. [Online]. Available:
  \url{https://doi.org/10.1007/s00220-008-0710-3}
\BIBentrySTDinterwordspacing

\bibitem{2022arXiv221102073I}
K.~{Ikeda} and A.~{Lowe}, ``{Quantum Protocol for Decision Making and Verifying
  Truthfulness among $N$-quantum Parties: Solution and Extension of the Quantum
  Coin Flipping Game},'' \emph{arXiv e-prints}, p. arXiv:2211.02073, Nov. 2022.

\bibitem{IKEDA2018199}
\BIBentryALTinterwordspacing
K.~Ikeda, ``Chapter seven - security and privacy of blockchain and quantum
  computation,'' in \emph{Blockchain Technology: Platforms, Tools and Use
  Cases}, ser. Advances in Computers.\hskip 1em plus 0.5em minus 0.4em\relax
  Elsevier, 2018, vol. 111, pp. 199--228. [Online]. Available:
  \url{https://www.sciencedirect.com/science/article/pii/S0065245818300160}
\BIBentrySTDinterwordspacing

\bibitem{10.1007/978-3-030-01174-1_58}
------, ``qbitcoin: A peer-to-peer quantum cash system,'' in \emph{Intelligent
  Computing}.\hskip 1em plus 0.5em minus 0.4em\relax Cham: Springer
  International Publishing, 2019, pp. 763--771.

\bibitem{814583}
J.~Watrous, ``Pspace has constant-round quantum interactive proof systems,'' in
  \emph{40th Annual Symposium on Foundations of Computer Science (Cat.
  No.99CB37039)}, 1999, pp. 112--119.

\bibitem{10.1145/335305.335387}
\BIBentryALTinterwordspacing
A.~Kitaev and J.~Watrous, ``Parallelization, amplification, and exponential
  time simulation of quantum interactive proof systems,'' in \emph{Proceedings
  of the Thirty-Second Annual ACM Symposium on Theory of Computing}, ser. STOC
  '00.\hskip 1em plus 0.5em minus 0.4em\relax New York, NY, USA: Association
  for Computing Machinery, 2000, p. 608–617. [Online]. Available:
  \url{https://doi.org/10.1145/335305.335387}
\BIBentrySTDinterwordspacing

\bibitem{10.1145/2049697.2049704}
\BIBentryALTinterwordspacing
R.~Jain, Z.~Ji, S.~Upadhyay, and J.~Watrous, ``Qip = pspace,'' \emph{J. ACM},
  vol.~58, no.~6, dec 2011. [Online]. Available:
  \url{https://doi.org/10.1145/2049697.2049704}
\BIBentrySTDinterwordspacing

\bibitem{PhysRevA.65.012310}
\BIBentryALTinterwordspacing
R.~W. Spekkens and T.~Rudolph, ``Degrees of concealment and bindingness in
  quantum bit commitment protocols,'' \emph{Phys. Rev. A}, vol.~65, p. 012310,
  Dec 2001. [Online]. Available:
  \url{https://link.aps.org/doi/10.1103/PhysRevA.65.012310}
\BIBentrySTDinterwordspacing

\bibitem{Marriott2005}
\BIBentryALTinterwordspacing
C.~Marriott and J.~Watrous, ``Quantum arthur--merlin games,''
  \emph{computational complexity}, vol.~14, no.~2, pp. 122--152, Jun 2005.
  [Online]. Available: \url{https://doi.org/10.1007/s00037-005-0194-x}
\BIBentrySTDinterwordspacing

\bibitem{gall2022distributed}
F.~L. Gall, M.~Miyamoto, and H.~Nishimura, ``Distributed quantum interactive
  proofs,'' \emph{arXiv preprint arXiv:2210.01390}, 2022.

\bibitem{55f26229cb9e4e719ad6aa994bc26328}
G.~Kol, R.~Oshman, and R.~Saxena, ``\BIBforeignlanguage{English
  (US)}{Interactive distributed proofs},'' in \emph{\BIBforeignlanguage{English
  (US)}{PODC 2018 - Proceedings of the 2018 ACM Symposium on Principles of
  Distributed Computing}}, ser. Proceedings of the Annual ACM Symposium on
  Principles of Distributed Computing.\hskip 1em plus 0.5em minus 0.4em\relax
  Association for Computing Machinery, Jul. 2018, pp. 255--264, 37th ACM
  SIGACT-SIGOPS Symposium on Principles of Distributed Computing, PODC 2018 ;
  Conference date: 23-07-2018 Through 27-07-2018.

\bibitem{10.1145/62212.62223}
\BIBentryALTinterwordspacing
M.~Ben-Or, S.~Goldwasser, J.~Kilian, and A.~Wigderson, ``Multi-prover
  interactive proofs: How to remove intractability assumptions,'' ser. STOC
  '88.\hskip 1em plus 0.5em minus 0.4em\relax New York, NY, USA: Association
  for Computing Machinery, 1988, p. 113–131. [Online]. Available:
  \url{https://doi.org/10.1145/62212.62223}
\BIBentrySTDinterwordspacing

\bibitem{Babai1991}
\BIBentryALTinterwordspacing
L.~Babai, L.~Fortnow, and C.~Lund, ``Non-deterministic exponential time has
  two-prover interactive protocols,'' \emph{computational complexity}, vol.~1,
  no.~1, pp. 3--40, Mar 1991. [Online]. Available:
  \url{https://doi.org/10.1007/BF01200056}
\BIBentrySTDinterwordspacing

\bibitem{10.1007/3-540-36136-7_11}
H.~Kobayashi and K.~Matsumoto, ``Quantum multi-prover interactive proof systems
  with limited prior entanglement,'' in \emph{Algorithms and Computation},
  P.~Bose and P.~Morin, Eds.\hskip 1em plus 0.5em minus 0.4em\relax Berlin,
  Heidelberg: Springer Berlin Heidelberg, 2002, pp. 115--127.

\bibitem{doi:10.1137/090751293}
\BIBentryALTinterwordspacing
J.~Kempe, H.~Kobayashi, K.~Matsumoto, B.~Toner, and T.~Vidick, ``Entangled
  games are hard to approximate,'' \emph{SIAM Journal on Computing}, vol.~40,
  no.~3, pp. 848--877, 2011. [Online]. Available:
  \url{https://doi.org/10.1137/090751293}
\BIBentrySTDinterwordspacing

\bibitem{2019arXiv190405870N}
A.~{Natarajan} and J.~{Wright}, ``{NEEXP in MIP*},'' \emph{arXiv e-prints}, p.
  arXiv:1904.05870, Apr. 2019.

\bibitem{reichardt2013classical}
B.~W. Reichardt, F.~Unger, and U.~Vazirani, ``Classical command of quantum
  systems,'' \emph{Nature}, vol. 496, no. 7446, pp. 456--460, 2013.

\bibitem{10.1145/22145.22192}
\BIBentryALTinterwordspacing
L.~Babai, ``Trading group theory for randomness,'' in \emph{Proceedings of the
  Seventeenth Annual ACM Symposium on Theory of Computing}, ser. STOC
  '85.\hskip 1em plus 0.5em minus 0.4em\relax New York, NY, USA: Association
  for Computing Machinery, 1985, p. 421–429. [Online]. Available:
  \url{https://doi.org/10.1145/22145.22192}
\BIBentrySTDinterwordspacing

\bibitem{ben2018scalable}
E.~Ben-Sasson, I.~Bentov, Y.~Horesh, and M.~Riabzev, ``Scalable, transparent,
  and post-quantum secure computational integrity,'' \emph{Cryptology ePrint
  Archive}, 2018.

\bibitem{9152704}
J.~Zhang, T.~Xie, Y.~Zhang, and D.~Song, ``Transparent polynomial delegation
  and its applications to zero knowledge proof,'' in \emph{2020 IEEE Symposium
  on Security and Privacy (SP)}, 2020, pp. 859--876.

\bibitem{10.1145/3133956.3134104}
\BIBentryALTinterwordspacing
S.~Ames, C.~Hazay, Y.~Ishai, and M.~Venkitasubramaniam, ``Ligero: Lightweight
  sublinear arguments without a trusted setup,'' in \emph{Proceedings of the
  2017 ACM SIGSAC Conference on Computer and Communications Security}, ser. CCS
  '17.\hskip 1em plus 0.5em minus 0.4em\relax New York, NY, USA: Association
  for Computing Machinery, 2017, p. 2087–2104. [Online]. Available:
  \url{https://doi.org/10.1145/3133956.3134104}
\BIBentrySTDinterwordspacing

\bibitem{Kiktenko_2018}
\BIBentryALTinterwordspacing
E.~O. Kiktenko, N.~O. Pozhar, M.~N. Anufriev, A.~S. Trushechkin, R.~R. Yunusov,
  Y.~V. Kurochkin, A.~I. Lvovsky, and A.~K. Fedorov, ``Quantum-secured
  blockchain,'' \emph{Quantum Science and Technology}, vol.~3, no.~3, p.
  035004, may 2018. [Online]. Available:
  \url{https://dx.doi.org/10.1088/2058-9565/aabc6b}
\BIBentrySTDinterwordspacing

\bibitem{1443098}
B.~Rosgen and J.~Watrous, ``On the hardness of distinguishing mixed-state
  quantum computations,'' in \emph{20th Annual IEEE Conference on Computational
  Complexity (CCC'05)}, 2005, pp. 344--354.

\bibitem{1181970}
J.~Watrous, ``Limits on the power of quantum statistical zero-knowledge,'' in
  \emph{The 43rd Annual IEEE Symposium on Foundations of Computer Science,
  2002. Proceedings.}, 2002, pp. 459--468.

\end{thebibliography}
\end{document}